\documentclass[11pt]{article}

\usepackage{booktabs} 
\usepackage{amsmath}
\usepackage{amssymb}
\usepackage{amsthm}
\usepackage[usenames, dvipsnames]{color}
\usepackage{url}
\usepackage{mathrsfs}
\usepackage{amsfonts}
\usepackage{graphicx}
\usepackage{textcomp}
\usepackage{multirow}
\usepackage{lipsum}

\usepackage[caption=true]{subfig}
\usepackage{algorithm}
\usepackage{algpseudocode}

\usepackage[margin=1in]{geometry}
\newtheorem{asm}{Assumption}
\newtheorem{dfn}{Definition}
\newtheorem*{dfn*}{Definition}
\newtheorem{lem}{Lemma}
\newtheorem{thm}{Theorem}
\newtheorem*{thm*}{Theorem}
\newtheorem{clm}{Claim}

\begin{document}
\title{Social Learning with Questions}
\author{
Grant Schoenebeck, Univeristy of Michigan, \textit{schoeneb@umich.edu} \\
\and Shih-Tang Su, Univeristy of Michigan, \textit{shihtang@umich.edu} \\
\and Vijay Subramanian, Univeristy of Michigan, \textit{vgsubram@umich.edu}}
\clearpage
\maketitle
\thispagestyle{empty}

\begin{abstract}
This work studies sequential social learning (also known as Bayesian observational learning), and how private communication can enable agents to avoid herding to the wrong action/state. Starting from the seminal BHW\footnote{Bikhchandani, Hirshleifer, and Welch, 1992 \cite{BHW1992}} model where asymptotic learning does not occur, we allow agents to ask private and finite questions to a bounded subset of their predecessors. While retaining the publicly observed history of the agents and their Bayes rationality from the BHW model, we further assume that both the ability to ask questions and the questions themselves are common knowledge. Then interpreting asking questions as partitioning information sets, we study whether asymptotic learning can be achieved with finite capacity questions. Restricting our attention to the network where every agent is only allowed to query her immediate predecessor, an explicit construction shows that a 1-bit question from each agent is enough to enable asymptotic learning.  

\end{abstract}
\clearpage
\pagenumbering{arabic}
\newpage

\section{Introduction} \label{sec:intro}
In social networks, agents use information from (a) private signals (knowledge) they have, (b) learning past agents actions (observations), and (c) comments from contactable experienced agents (experts) to guide their own decisions. The study of learning using private signals and past agents' actions, i.e., private or public history, in homogeneous and Bayes-rational agents was initiated by the seminal work in \cite{BHW1992,welch1992,banerjee1992}. Key results in \cite{BHW1992,welch1992,banerjee1992} state that in the model where countably infinite Bayes-rational agents make decisions sequentially to match binary unknown states of the world, named social learning or Bayesian observational learning in the literature \cite{smith2000,sorensen1996rational,hann2018speed,acemoglu2017fast,
acemoglu2011bayesian,lee1993convergence,gul1995endogenous,
vives1997learning,huck1998informational}, an outcome called \textit{Information Cascade} occurs almost surely with fully observable history and bounded likelihood ratio of signals. An information cascade occurs when it is optimal for agents to ignore their own (private) signals for decision making after observing the history. Though individually optimal, this may lead to a socially sub-optimal outcome where after an information cascade, all agents ignore their private signal and choose the wrong action, referred to as a “wrong cascade.” In models of social learning, there is another possible outcome known as \textit{(asymptotic) learning}\footnote{See Definition \ref{def:AL} in Section \ref{sec:Problem}.}. This occurs when the information in the private signals is aggregated efficiently so that agents eventually learn the underlying true state the world and make socially optimal decisions.

Literature studying social learning, i.e., making socially optimal decisions, is mainly focussed on using channels (a) and (b) above for Bayes-rational agents; works using  channel (c) mostly study learning by modeling it as a communication channel in distributed (sequential) hypothesis testing problems, but with non-Bayes-rational agents. Inspired by behaviors in social networks, where people usually query their friends who have prior experience before making their own decisions, using information via (c) by communicating with contactable past agents may reveal another channel to achieve asymptotic learning for Bayes-rational agents. Another reason to explore this channel is the following statement by Gul and Lundholm in \cite{gul1995endogenous}: ``\textit{Informational cascades can be eliminated by enriching the setting in a way that allows prior agents' information to be transmitted.}" This general principle, however, does not reveal whether learning can be achieved with finite-bit questions\footnote{Appendix~\ref{warm-up} shows a learning example where the total bits agent $n$ spends on questions is not bounded.}.

\vspace{6.0pt}
\noindent\textbf{Problem Statement:} Motivated by behaviors in real social networks and the quoted statement in \cite{gul1995endogenous}, we want to study if querying past agents with bounded number of finite capacity questions sequentially, but without relaxing assumptions of (a) and (b) used in BHW model \cite{BHW1992}, can achieve asymptotic learning or not. More precisely, we assume the essential features of the problem of sequential social learning with a common database recording actions of agents \cite{BHW1992,welch1992,banerjee1992}, but allow each Bayes-rational agent to ask a single, private and finite-capacity (for response) question of each among a subset of past agents (friends or assumed experts). Note that the Maximum Aposteriori Probability (MAP) rule is still individually optimally and will be used by each agent for her decision. We emphasize here that the BHW model \cite{BHW1992} has private signals with a bounded likelihood ratio, and with Bayes-rational agents only the ideas in Theme 2 are known to yield the learning outcome. In this paper, allowing for private questions, we want to answer the following three questions: 1) \textit{What is the minimum set of agents that need to be queried as a function of the agent's index and information set?}; 2) \textit{What is the minimum size of questions needed to achieve learning?}, and 3) \textit{Can we construct a set of questions that will achieve either minimum?}

	Before stating our main contributions, we highlight the main differences between this work and the existing literature. In the literature, there are three well-developed themes to achieve social learning: one, the presence of many agents with rich private signals (with unbounded likelihood ratios); two, the obfuscation of the history of agents' actions when viewed so that at least a minimal number of agents have their private signal revealed by their actions; and three, the presence of many (or all) non-Bayesian algorithmic agents. We briefly describe\footnote{See Related Work in Appendix~\ref{App:relatedwork} for a detailed discussion.} these three themes in the following:
	
	\noindent \textbf{\textit{Theme 1:}} Compared to the seminal BHW model \cite{BHW1992}, models in this theme use information from (b) and generalize the information content in (a). Unlike BHW and our model considering private signals with bounded likelihood ratios, the seminal work by Smith and S{\o}rensen \cite{smith2000} allowed generalized models for richer signals characterized by the (unbounded) likelihood ratio of the two states deduced from the private signals; and learning is achievable under richer signals.
	
	\noindent \textbf{\textit{Theme 2:}} Unlike BHW and our model disclose full history to every agent, works in this theme use partial history in (b). By revealing a random or a judiciously designed deterministic subset of past agents' actions in networks, \cite{acemoglu2011bayesian} and \cite{mossel2015strategic} respectively showed the asymptotic learning can be achieved through revealing either a random subset of history or at most $d$ past agents in a special class of networks. However, presenting reduced or distorted views of the history of agents' actions is philosophically troubling as it implicitly assumes that the distortions made are always benign or for efficiency, and also posits an implicit trust in the system designer on the part of the agents.
	
	\noindent \textbf{\textit{Theme 3:}} Here, by allowing non-Bayesian agents or changing payoff structures, there is a class of literature on distributed binary hypothesis testing problems that follows Cover \cite{cover1969hypothesis} and Hellman and Cover \cite{hellman1970learning} and falls under (c). In this model, agents can only observe the actions from their immediate predecessor, and their actions now try to transmit information and to learn collaboratively the true state of the world. It is shown in \cite{hellman1970learning} that the learning, often called optimal decision rules in this literature, cannot be achieved under bounded likelihood ratio of signals. However, if observing $K\geq 2$ immediate predecessor is allowed, authors in \cite{drakopoulos2012learning} showed that asymptotic learning can be achieved using a specific set of four-state Markov chains. From the perspective of information design, this approach of designing Markov chains for learning is similar in spirit to partitioning of information sets, but for non-Bayesian agents.
	
\vspace{1.0pt}
\noindent\textbf{Main Contributions:} Our analysis of the modified BHW model \cite{BHW1992} described above yields two main contributions: \\
1) To the best of our knowledge, we are the first to highlight the ability to change information structures among agents to achieve learning in social learning problems. The approach used in this work, namely partitioning information sets, is closely related to the Bayesian persuasion \cite{kamenica2011bayesian, rayo2010optimal}. Designs achieving asymptotic learning in our work can be viewed as a ``relayed Bayesian persuasion" by persuading agents in ``possibly wrong cascades" to avoid the information cascades eventually.\\
2) With an explicit construction of 1-bit question corresponding to the agent's possible information set and index value in the network and where agents are only allowed to query their immediate predecessors, we show that learning is achievable and answer the three questions addressed above in a single construction.

Note that in our approach, the system designer commits to a specific information structure (full public history of previous agents' actions plus private communications) without any reductions or distortions, and also provides the agents with a question guidebook, whose performance each agent can verify independently. The minimal nature of our learning achieving question guidebook also reveals the fragility of information cascades (Sec.16 in \cite{easley2012networks}), as a small amount of strategically delivered information leads to learning. The information revelation in our question guidebook is strategic in contrast to reviews that are \cite{le2016quantifying} generated via an exogenous process (and revealed only for some specific actions), and furthermore, lead to learning \cite{acemoglu2017fast} only if the signals are rich. A subtle but rather interesting point of our approach is the ``relayed persuasion" aspect wherein we only aim to persuade particular agents chosen by our design instead of agents chosen by nature as is commonly seen in Bayesian persuasion \cite{hedlund2017bayesian, ely2017moving}.

\section{Problem Formulation} \label{sec:Problem}
\textbf{BHW Model and Information Cascade}      
     Starting from the seminal BHW model\cite{BHW1992}, we consider a model with binary states of the world $\Theta=\{A,B\}$ that are equally likely to occur and a countable number of Bayes-rational agents, each taking a single action sequentially as indexed by $n \in \{1,2,...\}$. At each time slot $n$, agent $n$ shows up and chooses an action $X_n=\{\bar{A},\bar{B}\}$ with the goal of matching the true state of the world. Formally, for every agent $n$, the utility function $u_n(\Theta,X_n)$ is defined as $u_n(A,\bar{A}) = u_n(B,\bar{B}) = 1$ and $u_n(A,\bar{B})=u_n(B,\bar{A})=0$.
 
     Before agent $n$ takes her action, she receives an informative but not revealing private signal $s_n$. Her private signal is received through a binary symmetric channel (BSC) with a time-invariant crossover probability\footnote{The time-invariance assumption is mainly for the ease of algebraic complexity. The assumption we need for the model is the crossover probability of agents is common knowledge.} 1\textminus $p$, where $p\in (0.5,1)$, $s_n\in\{A,B\}$ for all $n$. An agent can also observe the full history of actions taken by her predecessors (agents with index lower than her, if any), $H_n \in \{\bar{A},\bar{B}\}^{n-1}$. The agent then computes the posterior belief of the true states of the world (alternatively, the likelihood ratio of one state versus the other), and takes the action corresponding to the most likely state. As in \cite{monzon2017aggregate,le2017information}, if indifferent between the two actions, we assume that the agent follows her private signal, instead of randomizing, following the majority, etc.
     \begin{dfn} \label{def:AL}
     	In a model of Bayesian observational learning, \textbf{asymptotic learning} (in probability) is achieved if $\lim_{n\to \infty}P(\{X_n=\bar{A}|\Theta=A\})=\lim_{n\to \infty}P(\{X_n=\bar{B}|\Theta=B\})=1$. If asymptotic learning is not achievable, we say that an \textbf{information cascade} has occurred.
     	\vspace{-6pt}  
     \end{dfn}
     
     Under the above setting, we know that the BHW model\cite{BHW1992} has an information cascade because all agents will ignore their private signals and imitate their immediate predecessor's action when the difference between the number of actions observed is greater than or equal to two\footnote{See Section 16.2 in \cite{easley2012networks}.}.
    
\vspace{1.0pt}
\noindent \textbf{Deterministic Network Topology with Finite Channel Capacity}         
     Agents' communication is modeled by a deterministic network $G(\mathbb{N},E)$ with directed edges. On each such edge, agents are allowed to transmit information up to pre-determined $K$ bits through a perfectly reliable communication channel. Since agent $m$ takes action prior to $n$ for all $m<n$, only directed edges $(m,n),m<n$ are allowed in the network. Using these additional directed edges, agent $n$ can ask questions individually and privately to predecessors in the set $\mathcal{B}_n=\{m|(m,n)\in E\}\subseteq [n-1]$ in line with the network topology, before making her decision. Agent $n$ asks the set of questions \textbf{after} receiving her private signal. Critically, the deterministic\footnote{Discussions about the difference between deterministic and randomized network on this problem are in Appendix~\ref{detvsranNW}.} network topology is common knowledge.

\vspace{1.0pt}
\noindent\textbf{Questions and Information Sets} Since the network topology is pre-determined, the set of agents that a particular agent can query for information is exogenous. Given the topology, we allow the information designer to supply the agents with questions that they can ask the set of contactable predecessors $B_n,~ n\in \mathbb{N}$ (in order to distinguish which information set they are at). Assume $m\in \mathcal{B}_n$, questions being asked by agent $n$ to agent $m$ are allowed to be dependent not only on the private signal and history observed by agent $n$, but also on the answers of questions asked to other agents in $\mathcal{B}_n$ prior to asking agent $m$. In short, the order that agents in $\mathcal{B}_n$ are queried in matters in the general framework. However, in this paper, we restrict our attention to only allow agents to ask questions simultaneously\footnote{The general framework is discussed in Appendix~\ref{extQGB}.} to predecessors in $\mathcal{B}_n$ to avoid complex analyses owing to the recursive analysis required to understand the engendered higher order beliefs. With this degeneration, such a collection of a set of questions is called a \textit{question guidebook (QGB)}. With this background in mind, we formally define a QGB. 
	
   
    \begin{dfn} \label{def:QB}
     A question guidebook $Q$ is a function $Q(n,H_n,s_n,m)$ that gives agent $n$ a set of predetermined questions conditioned on agent $n$'s private signal, the observed history, and the predecessor agent $m$ who agent $n$ queries. 
     \end{dfn}
     We assume that agents being queried are truthful in their responses, which is natural as they cannot gain from lying since their (payoff relevant) action has already been taken. 
    
    A QGB should have two important properties: feasibility and incentive compatibility. Intuitively, a QGB is feasible if agent $n$ only asks questions that she knows queried agent $m<n$ can answer.  This avoids the ambiguity of what would happen if agent $m$ does not have the information to answer the question asked of it. Formally, we call an observed history $H_n$ \emph{feasible} for a guidebook $Q$ if there is a nonzero probability of obtaining this history given all of agents follow the guidebook $Q$, that is $\mathbb{P}(H_m|Q\text{ is followed})>0$, $\forall m<n$; to avoid cumbersome notation we will simplify this to $\mathbb{P}(H_m|Q)$. Similarly, we call a question asked of agent $m$ by agent $n>m$ \emph{feasible} if for every $s_n$ and feasible history $H_n$, the question can be answered by agent $k$ using only the information that it possesses.  
    
    \begin{dfn}
     A question guidebook $Q$ is \textbf{feasible} if for every agent $n$, under every feasible observed history $H_n$\footnote{Since the history is fully observable in our model, $H_i\subset H_j$ for all $i<j$, it is sufficient to check $\mathbb{P}(H_{n-1}|Q)>0$ for the feasibility of history $H_n$. However, the above definition still works when the history is only partial observable.} and $s_n$, all questions provided by $Q$ are feasible.          
     \vspace{-6pt}
     \end{dfn}     
     
    A QGB is \emph{incentive compatible} if agent $n$ always asks a question from the set of questions that maximizes her expected payoff. 

     \begin{dfn} 
     A question guidebook $Q$ is \textbf{incentive compatible} if for every agent $n>1$, under every feasible observed history $H_n$ and $s_n$, the set of questions provided to her maximizes her expected utility among all feasible questions she can ask to agents in $\mathcal{B}_n$. \vspace{-6pt}
     \end{dfn}
     
     Given a feasible and incentive compatible QGB, each question serves to partition sequences of private signals to information sets: the possible states of randomness (underlying $\sigma$-algebra) into subsets (sub-$\sigma$-algebras to be more precise). Note that the answer to such a question specifies the set where the current realization belongs to. Viewing questions from the perspective of partitions helps in justifying the following assumption.

\begin{asm} \label{asm1} 
	We assume that there is no cost in asking questions. Thus, if no feasible question will change the expected payoff of an agent, there is no restriction for an information designer to design any question guidebook demands that this agent asks any provided questions within the prescribed capacity limit. 
	\vspace{-6pt}
\end{asm}     
     
	Assumption \ref{asm1} stands when a limited number of questions are asked at no cost; and every agent is willing to help her successors, in essence bringing in some level of cooperation. The reason for this assumption is because even though no questions can benefit the current agent, her information could be beneficial to future agents. This is inspired by behavior in social networks, where is common to see people are willing to help their friend/neighbor nodes. Henceforth, we will assume such behavior. Additionally, we will also assume that the QGB is common knowledge\footnote{A subtle weaker assumption is discussed in Appendix \ref{commknown}}.

\vspace{-12pt}
\section{Telephone-game Network, Strategy, and Question Guidebooks} \label{sec:Markov}
\vspace{-6pt}
     To exhibit how designing QGBs can achieve learning, in the following sections, we consider a special network, called telephone-game network. In this network, we only endow each pair of consecutive agents with a communication channel of finite capacity. In other words, the topology of this communication network is a directed line graph. Therefore, besides observing the actions in history, the only source for an agent to get additional information is asking finite-bit questions to her immediate predecessor (if any). Since the capacity of every channel between each agent and her immediate predecessors is finite, agents may not be able to get the information of all private signals observed by all her predecessors. This is a very basic model\footnote{The telephone-game network has the same topology as a tandem network, but the literature on learning in tandem networks uses one-way communication made before the next agent sees her private signal. In contrast, our questions are conditional on the received private signals. To avoid any confusion, we avoid the tandem network terminology.} to start studying if the asymptotic learning is achievable by finite-bit questions.
      In this Section, we will first propose a high-level strategy that may achieve asymptotic learning, then study the following two important questions: first, how to design QGBs to get and accumulate the information we want, we will use the term \textit{Information set partition} in the following paragraphs; second, what are the necessary and sufficient conditions to implement the strategy we proposed to achieve learning?
    
    \vspace{-12pt} 
     \subsection{Threshold-based Strategy} \label{TBstrategy}
     \vspace{-6pt}
     Capacity constraints may make agents incapable of getting the information of all their predecessors' private signals. However, there are some ``clues'' that agents can learn from a fixed length sequence of private signals owing to the different signal distributions among states of the world, e.g., when the true state is $A$, private-signal sequences containing three consecutive $A$s will occur more frequently than when the true state is $B$. With this intuition in mind, we propose the following strategy.
        
   Suppose we are in an $\bar{A}$ cascade, i.e., every following agent ignores her private signal and takes action $\bar{A}$ if no information from questions is provided. As we all know, if the true state is $B$, agents are less likely to get multiple private signal $A$s in a row. Hence, we propose the following strategy:\\
\textbf{\textit{Step 1:}} Given a predetermined threshold $K$, agents in an A cascade ask questions to know if the event ``$K$ consecutive private $A$ signal are received by agents'' occurs or not.\\
\textbf{\textit{Step 2:}} Using the information of whether the event occurs or not, an agent updates her likelihood ratio of state $B$ versus $A$. Starting from agent $m$ with a fixed likelihood ratio of $B$ versus $A$, there will be the first agent $n>m$ with a positive probability that her updated likelihood ratio $B$ versus $A$ will cross $1$.\\
\textbf{\textit{Step3:}} If agent $n$ has likelihood ratio $B$ versus $A$ $\geq 1$, her best strategy is to stop the cascade. If she doesn't stop the cascade, then we use this agent $n$ as a new starting agent (agent $m$ in previous step) and start a new round checking if the event occurs in the following agents. 
     
	Let's defer the details of the implementation of this strategy to the end of this Section and to Section \ref{sec:DetQB}. Suppose there is a way to implement the threshold-based strategy presented above without any incentive compatibility issues, then an important feature we can observe here is that not every agent in a cascade gets the chance to stop the cascade, i.e., some of them just ask questions and forward the information of whether the event occurs or not to their successor. Note that Assumption \ref{asm1} grants the flexibility of designing a QGB with these questions. 
	
	 Prior to designing QGBs to implement the threshold-based strategy, we know that for agents who have a chance to stop the cascade conditional on the observed history, the questions designed for them are restricted to the set of questions that maximize their expected utility. However, for agents indifferent to any questions, we can design any question for them. Thus, in order to design questions systematically, we categorize agents into two types conditional on the QGB in practice and the observed history. 
     \begin{dfn} \label{def:agenttype} \vspace{-6pt}
        Given an observed history in a feasible and incentive compatible QGB, an agent is \textbf{active} if she has a positive probability to stop the cascade, otherwise she is called \textbf{silent}.
        \vspace{-12pt}
     \end{dfn}          
   To further clarify the above definition, conditional on the history, agents who may benefit from the answer of questions are active agents, otherwise they are silent.
   
   \subsection{Questions and Information Set Partition} \label{InfoSet}
   \vspace{-6pt}
	Here, we detail how questions help gather information, and how the capacity constraints limit the (lossless) information aggregation. To begin, let $T_n$ be the set of sequences of private signals from agent $1$ to $n$\textminus 1, $T_n \ni (s_i)_{i=1}^{n-1}$. The information space of agent $n$, corresponding to a observed history $H_n$ and the question guidebook $Q$, is the set of all feasible sequences $T^*_n(Q,H_n) \subseteq T_n$. The questions assigned to agent $n$ help her to update her posterior belief of the true state by partitioning $T^*_n$ into information sets $I_n(Q,H_n)$ and telling agent $n$ at which set she is in. For simplicity of notation, we denote the collection of information sets $I_n(Q,H_n)$ by $\mathcal{I}_n(Q,H_n)$.
	
	By viewing question as an information set partition tool, studying the transition from $\mathcal{I}_n(Q,H_n)$ to $\mathcal{I}_{n+1}(Q,H_{n+1})$ can help us understand how information aggregates. A well-known approach to analyze this class of problems systematically is by mapping these information-set transitions to Markov chains with transition matrices\footnote{To allay readers' concerns on what typical questions should look like and why we use Markov chains to model the QGBs, an example demonstrating the delicateness of the design problem is in Appendix~\ref{sec:delicate}.}. To formulate the mapping, we associate a QGB with a sequence of sets of Markov chains sharing the same state space. Denote $(G,(\mathcal{M}_n))$ to be a sequence of sets of Markov chains, where $G$ is the state set and $\mathcal{M}_n$ is the set of Markov chains at time $n$. To represent all distinguishable information sets, an inequality  $|G|\geq \sup_{n,H_n\in \mathcal{H},s_n} |\mathcal{I}_n(Q,H_n)|$ is required for the corresponding $(G,(\mathcal{M}_n))$. Since state space $G$ is shared by all $\mathcal{M}_n$, hereafter, we simplify words ``agent knows her information set corresponding to state $G_i$" to ``agent goes to $G_i$."

\vspace{-12pt}    
    \subsection{Conditions for Asymptotic Learning}
    \vspace{-6pt}
		In the threshold-based strategy proposed above, with the threshold is fixed and supposing that we are in a cascade, the number of silent agents between two active agents (if any)\footnote{We show that active agents cannot come consecutively in the proof of Claim \ref{lem1} in Appendix~\ref{pf:lem2}} is only determined by the (prior) likelihood ratio of the first silent agent who arrives right after an active agent. Denote the number of silent agents between the $k$-th and $k+1$-th active agents by $m^k$. One necessary condition to achieve asymptotic learning is having $m^k$ go to infinity with $k$. If this does not happen, then we always have a positive probability (lower-bounded by a constant) to stop a cascade, which will stop every correct cascade too. However, if $m^k$ goes to infinity too fast, then we cannot guarantee that a wrong cascade will eventually be stopped. Ergo, to achieve asymptotic learning, the key is to choose questions that carefully control the growth rate of $m^k$. In the next Section, we will detail an implementation only using a $1$-bit question for each agent. The precise necessary and sufficient conditions to achieve learning for threshold-based QGBs will be presented in Lemma 2.

\vspace{-12pt}
\section{Implementation of Threshold-based Strategy -- Asymptotic \\Learning achieved in One-bit Questions} \label{sec:DetQB}
\vspace{-6pt}
    In this section, we present our main result: asymptotic learning is achievable via an explicit construction of a question guidebook that uses a threshold-based strategy in the telephone-game network where agents ask exactly a single one-bit question.
    
    \begin{thm}\label{thm:main} 
        In the telephone-game network, there exists a question guidebook using 1-bit capacity questions that achieves asymptotic learning.
    \end{thm}
    \vspace{-6pt}
    Before presenting the proof, we first construct the QGB, and argue the feasibility and incentive compatibility of the designed QGB. Then, we provide the necessary and sufficient conditions when this class of QGBs achieves asymptotic learning. Finally, we prove the theorem by showing the conditions are satisfied in the constructed QGB.
    
\noindent\textbf{Construction of the Question Guidebook}    
	 Here we implement a QGB using the threshold-based strategy introduced in Section \ref{TBstrategy}. Under the constraint of 1-bit capacity  on questions, we choose $K=2$. In other words, in the QGB we are going to construct, we know that once two consecutive silent agents get a private signal of the same type as the current cascade, there is no chance to stop the cascade at the immediate next active agent.         
        
        Without loss of generality we assume that the true state of the world is $B$ so that an $\bar{A}$ cascade is a wrong cascade and a $\bar{B}$ cascade is a correct one. Moreover, to avoid trivial questions in the QGB, \textit{the question guidebook becomes operational only when a cascade is initiated}. Thus, an agent not in a cascade uses her private signal and does not need the guidebook.
        
      As described in Section \ref{TBstrategy}, we design different questions for active and silent agents in a cascade. Since only 1-bit questions are allowed, Section \ref{InfoSet} and accompanying details in Appendix~\ref{QGBtoMC} suggest that the QGB will consist of the partition of the (evolving) history into five sets. Next we illustrate the QGB by providing the corresponding Markov chains first and then detailing the questions actually being asked in each possible state. The Markov chains are depicted in Figure \ref{fig1} assuming that a $\bar{A}$ cascade is underway\footnote{In a $\bar{B}$ cascade, the same guidebook applies but with the $A$s and $B$s swapped.}, and they prescribe how the information space gets partitioned based on the type of the agent (active or silent), the private signal of the agent and the response from the immediate predecessor to the question (if it is asked). 
    \begin{figure}[H] 
    \centering
	\includegraphics[width=0.55\linewidth]{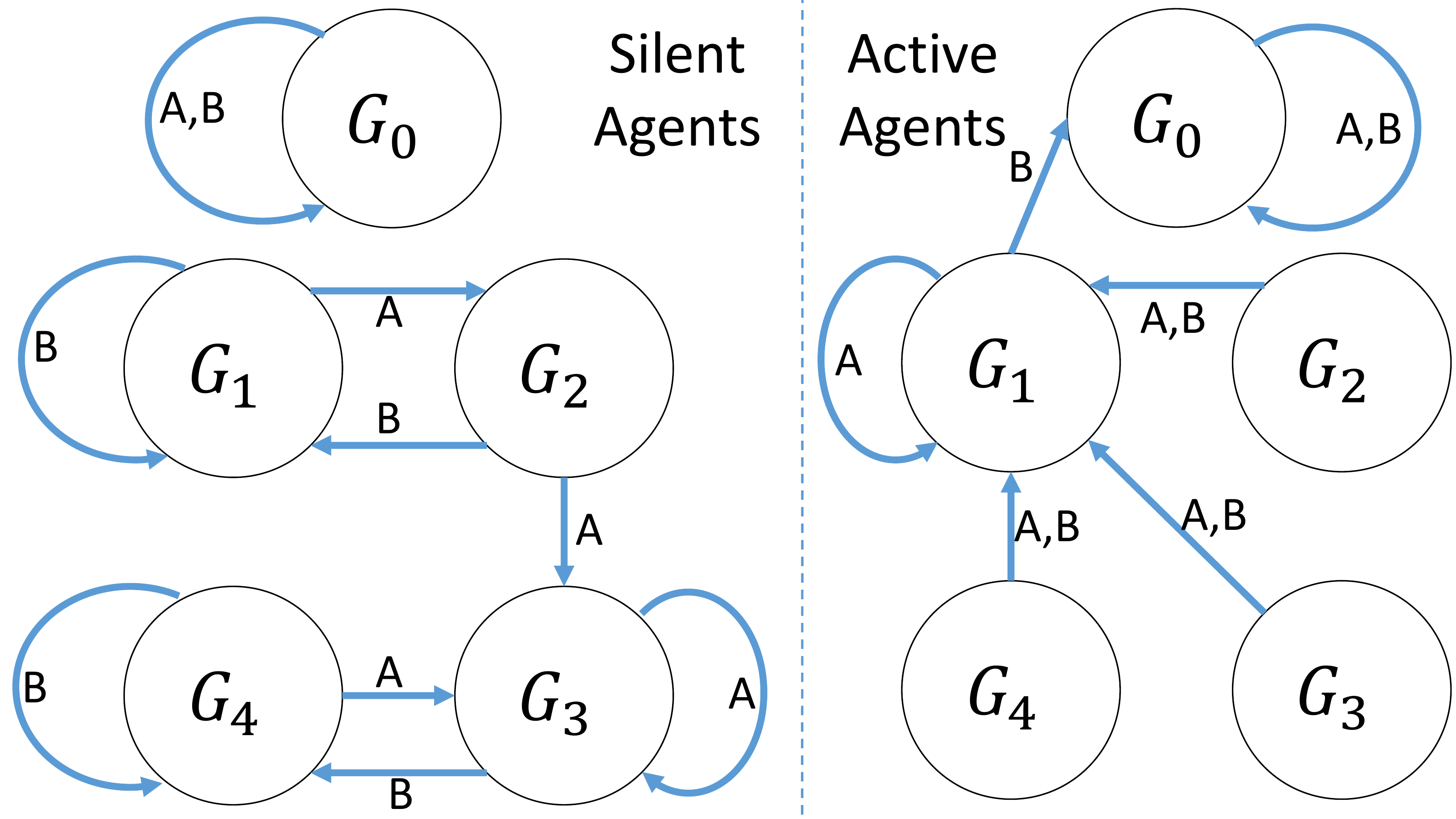}
    \caption{Markov chains of proposed threshold-based question guidebook}
    \label{fig1}
    \end{figure} \vspace{-12pt}         
        The corresponding Markov chains in the designed QGB, as depicted in the left of Figure \ref{fig1}, endows every silent agent with the same transition matrix. This transition matrix is given by two different determined questions conditional on the silent agent's private signal\footnote{As described earlier, a silent agent's partition can never be $G_0$.}; and the questions and corresponding information sets are as follows.
\vspace{-6pt}        
\begin{table}[H]
\begin{tabular}{|l|l|l|}
\hline
 &  Receives private signal \textbf{B} & Receives private signal \textbf{A} \\ \hline
Question asked & Are you in $\{G_1, G_2\}$? & Are you at $G_1$? \\ \hline
Action under positive answer& Go to $G_1$  & Go to $G_2$\\ \hline
Action under negative answer& Go to $G_4$ & Go to $G_3$\\ \hline
\end{tabular}
\end{table}     
\vspace{-6pt}          
        
        
        Since every active agent only cares if her immediate predecessor is in $G_1$ when she receives a private signal $B$ so she can stop the cascade (play $\bar{B}$), questions are only needed in that case.

\vspace{-8pt}  
\begin{table}[H]
\begin{tabular}{|l|l|l|}
\hline
 &  Receives private signal \textbf{B} & Receives private signal \textbf{A} \\ \hline
Question asked & Are you at $G_1$? & No questions asked \\ \hline
Action under positive answer& Go to $G_0$ and stop cascade & \multirow{2}{*}{Go to $G_1$}\\ \cline{1-2}
Action under negative answer& Go to $G_1$ &\\ \hline
\end{tabular}
\end{table}         
        
\vspace{-6pt}          
        
        
        Before showing that the QGB is feasible and incentive compatible, we want to point out that while there exist a large set of QGBs that can achieve asymptotic learning even with the $1$-bit constraint, the proposed design simplifies the analysis and proofs, and avoids solving complex recursive system of equations with four variables.
    
                
\noindent \textbf{Feasibility and incentive compatibility of the question guidebook}        With the proposed QGB in hand, the first step is to verify the feasibility and incentive compatibility. Feasibility of this QGB is straightforward because every agent, no matter whether she is active or silent, knows her current state. Therefore, she can definitely answer the yes-no question about her current state to pass the feasibility check. 
        
        Showing incentive compatibility is equivalent to showing that $G_1$ is the only state that can have likelihood ratio of B over A crossing $1$ for any active agent. Here, we will establish a result that applies to more generally than the designed QGB. We will prove that every \textit{threshold-based} QGB has positive probability to stop the cascade only when the immediate predecessor of a active agent is at $G_1$ (threshold event not holding for current majority, i.e., $\bar{A}$).   

\vspace{-6pt}
        \begin{dfn} \label{dfn5}
        Given a question guidebook such that every silent agent uses the same transition matrix, a question guidebook is \textbf{threshold-based} if for every silent agent whose neighbor is in a transient state (e.g., $G_1, G_2$ in this question guidebook) of the Markov chain, she goes to state $G_{i+1}$ upon receiving a private signal in an observed majority or goes to $G_{\max \{1,i-1 \}}$ upon receiving a private signal in an observed minority. Furthermore, active agents continue the cascade bring all feasible sequence of private signal to $G_1$.
        \end{dfn}         
        \begin{lem} \label{lem2} 
        In threshold-based question guidebooks, active agents can only stop the cascade at $G_1$.
        \end{lem}
\vspace{-6pt}        
        To show Lemma \ref{lem2}, we first need to guarantee there exists at least one silent agent between any pair of active agents, i.e., active agents cannot arrive consecutively. The idea of the proof is that once an active agent fails to stop a cascade, it either gets an observed-majority private signal or the cascade will continue whatever private signal she gets. Then simple algebra rules out the possibility of consecutive active agents: see Claim \ref{lem1} in Appendix \ref{pf:lem1}. Now, given the current active agent $a^k$, if the next active agent indexed $a^{k+1}$ can stop cascade at state $G_i$ for some $i>1$, then agent $a^{k+1}-i+1$ also has the ability to stop the cascade at $G_1$, which contradicts the fact that $a^{k+1}$ is the next active agent. The detailed proof is in Appendix \ref{pf:lem2}.
        
        Since the QGB we constructed is a threshold-based QGB and active agents stop a cascade only in $G_1$, Lemma \ref{lem2} guarantees incentive compatibility.
   
   \vspace{2pt}      
\noindent \textbf{Necessary and sufficient conditions for asymptotic learning in threshold-based question guidebooks}
    In order to show that the proposed QGB can achieve asymptotic learning, we provide a necessary and sufficient condition to achieve asymptotic learning for every threshold-based QGB.

	\begin{dfn} \vspace{-6pt} 
	Let $h^+(m)$ be the probability that the wrong cascade will be stopped when $m^k=m$, where $m^k$ is the number of silent agent between the $(k-1)^{th}$ and $k^{th}$ active agents. Similarly, $h^-(m)$ represents the probability that the right cascade will be stopped when $m^k=m$.
	\end{dfn}

    \begin{lem} \label{lem3} 
    Given a threshold-based question guidebook $Q$ that is operational in a cascade. The following three conditions are necessary and sufficient for the question guidebook to achieve to achieve the asymptotic learning:
    \begin{itemize}
    \vspace{-6pt}
    \item[1] $\lim_{k\rightarrow\infty} m^k=\infty$;
    \vspace{-6pt}
    \item[2] The growth rate of $m^k$ satisfies $\Pi_{k=1}^\infty(1-h^-(m^k))>\Pi_{k=1}^\infty(1-h^+(m^k))=0$;
    \vspace{-6pt}
    \item[3] The transition matrix $M^*=M_sM_a$ is irreducible, where $M_s$ is the transition matrix for silent agents and $M_a$ is the transition matrix for active agents who receive observed majority signals.
    \end{itemize}
\end{lem}
\vspace{-6pt} 
    The first condition makes the frequency of active agents to go to 0 as time goes to infinity, otherwise asymptotic learning cannot be achieved\footnote{If we have a non-zero proportion of agents that take actions according to their private signals, the probability of the right action is upper bounded away from 1.}. Furthermore, we want every wrong cascade to be stopped with probability 1, but also need the right cascade to have a positive probability to last forever. This constrains the maximum and minimum growth rate of $m^k$, which is discussed in the second condition. The last condition guarantees that in a cascade, with an arbitrary number of silent agents followed by an active agent, all states can be visited. Without this condition, the QGB cannot correct all kinds of the wrong cascades and no learning is achievable. 

\vspace{-6pt}    
\subsection{Asymptotic Learning is achieved (Sketch of proof of Theorem 1)} 
\vspace{-6pt}
    Since the last condition in Lemma \ref{lem3} is trivially satisfied in the designed QGB, most of proof is on verifying the first two conditions in Lemma \ref{lem3}. For this we have to analyze the growth rate of $m^k$ thoroughly. This following paragraphs will first characterize the form of $m^k$, then study upper and lower bounds of its growth rate. With those bounds in hand, calculations can be done to show the wrong cascade will be stopped almost surely and the right state will last forever with positive probability (which can be lower-bounded by some constant).

\noindent \textbf{Form of number of silent agents}  
    To study the growth rate of $m^k$, we need to know exactly how many of agents between each pair of consecutive active agents. 
       
    Prior to this, we need to specify the functions $h^+(m)$ and $h^-(m)$ first. With threshold $K=2$, two consecutive majority signals in a cascade will continue the cascade at the next active agent. Simple combinatorics yields $h^+(m)$ and $h^-(m)$ as follows:
  \vspace{-6pt}  \begin{eqnarray}\label{eq:hform}
 \qquad h^+(m)=\sum_{i=0}^{m/2} {m-i-1 \choose i}p^{m+1-i}(1-p)^i; \qquad
    h^-(m)=\sum_{i=0}^{m/2} {m-i-1 \choose i}p^i(1-p)^{m+1-i}. 
\end{eqnarray}     
Furthermore, with Lemma \ref{lem2}, we know that the likelihood ratio of $B$ versus $A$
	at state $G_1$ is the only parameter that the next active agent needs to compute. Suppose we know the likelihood ratio of the first silent agent right after $a^k$, the next active agent $a^{k+1}$ is the first agent who could have likelihood ratio at $G_1$ crossing $1$. Since $h^+(m)$($h^-(m)$) is the probability that an agent and her $m^{th}$ successor are both at $G_1$ conditional on the right(wrong) cascade not yet stopped. The ratio of $h^+(m)$ over $h^-(m)$ is the likelihood ratio of $B$ versus $A$ conditioned on the event that the cascade continues after $m$ silent agents. Hence, by definition of a silent agent, if the agent with index $a^{k-1}+m$ is silent, she must have likelihood ratio at $G_1$, $\ell_{G_1}^{H_{a^{k-1}+1}}\frac{h^+(m)}{h^-(m)}<1$. Given the fact that $\frac{h^+(m)}{h^-(m)}$ is an strictly increasing function of $m$ for all $p>0.5$, the number of silent agents $m^k$ between active agents $a^k$ to $a^{k+1}$ can be mathematically defined as:
	\vspace{-6pt}
	\begin{eqnarray} \label{eq:lklratio}
	 \qquad \qquad \qquad m^k \equiv \min \{m|\ell_{G_1}^{H_{a^k+1}}\frac{h^+(m)}{h^-(m)}\geq 1\}. \vspace{-6pt}
	\end{eqnarray}	
	
	\vspace{-6pt} Then, since every active agent failing to stop the cascade has only one information set corresponding to $G_1$, there is a simple recursive form of likelihood ratio at state $G_1$ between first silent agents after $k^{th}$ and $(k+1)^{th}$ active agents by using the ratio of probability that the cascade continuous. The recursive form of $\ell_{G_1}^{H_{a^{k+1}+1}}$ a strictly decreasing\footnote{See Appendix \ref{pf:clm2} for the proof.} function of $k$, is given by:
	\vspace{-6pt}
	\begin{equation} \label{eqn:ldecrease}
	    \ell_{G_1}^{H_{a^{k+1}+1}}=\ell_{G_1}^{H_{a^{k}+1}}\frac{1-h^+(m^k)}{1-h^-(m^k)}
	\end{equation} 
	
	\vspace{-12pt} With functions $h^+(m)$, $h^-(m)$, and \eqref{eqn:ldecrease}, $m^k$ is non-decreasing and can be computed iteratively.

\noindent \textbf{Upper bound of growth rate of $m^k$}
    	If asymptotic learning can be achieved under this QGB, every wrong cascade must be stopped, i.e., $\Pi_{k=1}^{\infty} (1-h^+(m^k))=0$. Thus, if we can find a sequence $w^k \geq m^k,~\forall k$, and $\Pi_{k=1}^{\infty} (1-h^+(w^k))=0$, then $h^+(m^k)\geq h^+(w^k)$ guarantees $\Pi_{k=1}^{\infty} (1-h^+(m^k))=0$. 
    	Finding the upper bound of the growth rate of $m^k$ is equivalent to finding the lower bound of sequence $\underbar{n}_s$ such that $m^{k-1}=s-1$ and $m^{k+i-1}=s$ for all $i\leq \underbar{n}_s$. Now, suppose $m^{k}=s$, $m^{k-1}=s-1$, we can calculate $\underbar{n}_s$ (See detailed calculations in \ref{calns}) to get 
    	\vspace{-6pt}\begin{eqnarray} \label{eqn_lowbdd}
    \left|\{m^t=s|\ell_{G_1}^{H_{a^{t-1}+1}}\frac{h^+(s)}{h^-(s)}\geq 1\}\right| \geq \underbar{n}_s > \ln\Big(\frac{h^+(s)h^-(s-1)}{h^-(s)h^+(s-1)}\Big)  \geq \frac{\ln(c_1(p))}{p^s+p^{2s}},
    	\end{eqnarray}
where $c_1(p)$ is only a function of $p$.    	
	
	
\vspace{2pt}
\noindent \textbf{Wrong cascade will be stopped almost surely}
	Taking the inequality (\ref{eqn_lowbdd}) into the condition 2 in Lemma \ref{lem3}, the probability that a wrong cascade will be stopped is\footnote{See Appendix \ref{subsec:wc_stop} for detailed calculation}
  \begin{eqnarray}
& \mathbb{P}(\text{A wrong cascade can be stopped})  =  1-\Pi_{k=1}^{\infty} (1-h^+(m^k))  \geq  1-\Pi_{k=1}^{\infty} (1-p^{m^k}) \nonumber\\
  &\geq 1-\Pi_{k=1}^{\min\{j|m^{j+1}=\underbar{J}\}} (1-p^{m^k})\exp\Big(-p\ln(c_1(p))\sum_{n=\underbar{J}}^{\infty} \frac{1}{1+p^n}\Big) =1
  \end{eqnarray}

\noindent \textbf{Lower bound of growth rate of $m^k$ and the probability of stopping a right cascade}
	Similarly, to show that $\mathbb{P}(\text{a right cascade will be stopped})<1$, we can show that a lower bound of $\Pi_{k=1}^{\infty} (1-h^-(m^k))$ is positive. Thus, we now need a lower bound of the growth rate of $m^k$. Using a similar technique as for \eqref{eqn_lowbdd}, we derive an opposite inequality for $\bar{n}_s$ in Appendix~\ref{lowns}, and then the probability that a right cascade will be stopped is
 \vspace{-6pt}
  \begin{equation}\label{lasteqn}  
 	\mathbb{P}(\text{A right cascade will be stopped})\leq  1-B(p)\exp\Big(\sum_{n=\bar{J}}^{\infty}\dfrac{1}{
    c_4(p)c_1(p)^n-1}\Big), 
  \end{equation} 
where $B(p)=c_2(p)e^{c_3(p)}\Pi_{k=1}^{\min\{j|m^{j+1}=\bar{J}\}} (1-h^-(k)).$ Since $\sum_{n=\bar{J}}^{\infty}\tfrac{1}{
    c_4(p)c_1(p)^n-1}$ converges by the ratio test, the RHS of \eqref{lasteqn} is less than $1$, so that a right cascade can last forever with a positive probability.

  Now, every wrong cascade will be stopped but there is a positive probability that a right cascade will continue forever, so the second condition is satisfied. Furthermore, since the lower bound and upper bound of growth rate $\bar{n}_s$ suggest a finite interval of $m^k=s$, $m^k$ goes to infinity as $k$ goes to infinity, so the first condition also holds. Thus, asymptotic learning can be achieved under this QGB. Once we're out of a cascade, agents use their private signals and that will initiate another cascade with a bias towards a right cascade. Nevertheless, every wrong cascade will be stopped in finite time and an unstoppable right cascade happens after finitely many stopped right cascades. Thus, the learning happens \emph{almost surely} instead of just \emph{in probability}, as in Definition~\ref{def:AL}, so that in finite time learning occurs.
\vspace{-12pt}

\section{Conclusion}
We have shown that in the sequential social learning model of BHW \cite{BHW1992}, agents avoid wrong cascades and achieve asymptotic learning by asking one well designed binary question to the preceeding agent.  To do this, we develop a question guidebook that the agents can use to ask questions such that the agent can always answer the question, and agents always ask a question that best serves their self interest. Determining the distribution of the time of learning or finding the best question guidebook to minimize some statistic of the time of learning is for future work. Generalizing from binary private signals to other discrete signals and then to the framework of \cite{smith2000}, or to consider multiple states of nature \cite{sorensen1996rational} is also for future work. 
	
\bibliographystyle{IEEEtran}
\bibliography{ref_social_learning}
\newpage
\appendix
\section{Appendix: Proofs and Calculations} 

\subsection{Proof of Lemma 1}\label{pf:lem2}
        \begin{proof}
	Prior to showing Lemma 1, we want to first claim a useful property in threshold-based question guidebooks.
	\begin{clm} \label{lem1}
            Given a threshold-based question guidebook which is feasible and incentive-compatible, active agents cannot arrive consecutively.
	\end{clm}     
\begin{proof} \label{pf:lem1}
       First we note that active agents arriving means that the cascade is still ongoing and the question guidebook operational. Then assume that active agent $a^k$ is in state $G_j \in \mathcal{G}^*$, where $\mathcal{G}^*$ is defined as the set of states at which she can stop the cascade.
            
            We know for sure that the likelihood ratio of agent $a^k+1$ at state $G_1$, denoted by $\ell_{G_1}^{a^k+1}$ is less than $\frac{1-p}{p}$.  The reason is that $G_1$ at $a^k+1$ contains two types of events. The first types is where $a^k$ is in $G_k \in \mathcal{G}^*$ but receiving an observed majority. Given $\ell_{G_i}^{a^k}<1$ for all $G_i\in \mathcal{G^*}$, $(\ell_{G_1}^{a^k+1}|a^k \text{ at } G_i \in \mathcal{G^*})<\frac{1-p}{p}$; note that agent $a^k$ does not stop the cascade. The other type is $a^k$ was at $G_i\not\in \mathcal{G}^*$. Here also we get $(\ell_{G_1}^{a^k+1}|a^k \text{ at } G_i \not\in \mathcal{G^*})<\frac{1-p}{p}$ because $a^k$ cannot stop the cascade even after receiving observed minority in any of those states. 
            
            Therefore, we can conclude that $\ell_{G_1}^{a^k+1}<\frac{1-p}{p}$, which guarantees that agent $a^k+1$ is a silent agent, and not an active agent.
\end{proof}        
       Given that every active agent goes to $G_1$ once knowing she cannot stop the cascade, Claim \ref{lem1} tells that the next agent must be a silent agent.
       With the claim, we can prove this lemma by contradiction. Starting from $G_1$, we assume the next active agent is $a^{k+1}$ who can stop the cascade at $G_i$ for some $i>1$. Now, consider that likelihood ratio of state $G_1$ at agent $a^{k+1}-i+1$, $\ell_{G_1}^{a^{k+1}-i+1}$, $\ell_{G_1}^{a^{k+1}-i+1}\geq \ell_{G_i}^{a^{k+1}}$; this is because of the allowed transitions in the Markov chain of silent agents. Agent $a^{k+1}$ can stop the cascade at $G_i$ and this requires $\ell_{G_i}^{a^{k+1}}\geq \frac{1-p}{p}$. It implies that $\ell_{G_1}^{a^{k+1}-i+1}\geq \frac{1-p}{p}$. Now, agent $a^{k+1}-i+1$ must be an active agent, which contradicts that next active agent is $a^{k+1}$.
\end{proof}

\subsection{Proof of Lemma 2}\label{pf:lem3}

First, we show that if one of these three condition fails, then asymptotic learning is not achievable.

The check for the first condition is straightforward, if it fails, then there exists an $N$ such that $m^k<N$ for all $k$. Now, $h^-(m^k)$ is lower-bounded by $h^-(N)>0$. As we all know, $\lim_{n \to \infty}(1-h^-(N))^n=0$, a correct cascade will always being stopped in finite time and asymptotic learning is not possible.

For the second condition, if $\Pi_{k=1}^\infty(1-h^+(m^k))=0\neq 0$, then there is a positive probability that a wrong cascade lasts forever. Obviously, asymptotic learning cannot be achieved. If $\Pi_{k=1}^\infty(1-h^-(m^k))=\Pi_{k=1}^\infty(1-h^+(m^k))=0$, then we will keep stopping every cascade whether it is a right cascade or not. Therefore, $\lim_{n\rightarrow\infty} \mathbb{P}(X_n=\bar{X}_n)<1$ because we will have a positive probability of either being in a wrong cascade or not being in a cascade.

Finally, if the product of transition matrix $M^*=M_sM_a$ is not irreducible, then there are some states that we either cannot access or from which we cannot go back to any other state. It implies that we cannot stop cascades at those states. Hence, it is necessary that $M^*$ is irreducible to achieve asymptotic learning.

In the other direction, suppose we have the first and the second conditions satisfied, then we know we will keep stopping some wrong cascades but not the correct cascades. With the third condition, we know that if a cascade, whether it's correct or wrong, isn't stopped by the upcoming active agent, then it has a positive probability of being stopped by at least one of the next $|G|$ active agents, because the directed graph corresponds to $M*$ is strongly connected. Now, since all wrong cascades can be stopped in finite time but some of right cascade will last forever, then asymptotic learning is achieved.

\subsection{Upper bound of growth rate of $m^k$} \label{calns}
	Finding the upper bound of $m^k$ is analogous to finding a sequence of $\underbar{n}_s$ such that 
	$$\left|\left\{m^k=s|\ell_{G_1}^{H_{a^{k-1}+1}}\frac{h^+(s)}{h^-(s)}\geq 1\right\}\right|\geq \underbar{n}_s \; \forall\ s\in \mathbb{N}.$$ The following calculations derive the value of $\underbar{n}_s$.
	
	First, the number of $m^k=s$ is equivalent to difference of indices between the $k^{th}$ active agent to the next active agent with index $j$ such that $\ell_{G_1}^{H_{a^{j}+1}}\frac{h^+(s)}{h^-(s)}< 1$. Hence, we get the following equation:   
    	\begin{align*}
\left|\left\{m^t=s|\ell_{G_1}^{H_{a^t+1}}\frac{h^+(s)}{h^-(s)}\geq 1, t\geq k\right\}\right|=\min\{n|\ell_{G_1}^{H_{a^{k+n}+1}}\frac{h^+(s)}{h^-(s)}< 1\}.
\end{align*}
 	Now, using the recursive formula of likelihood ratio stated in \eqref{eqn:ldecrease}, we get the following equation:
\begin{align*}
\left|\left\{m^t=s|\ell_{G_1}^{H_{a^t+1}}\frac{h^+(s)}{h^-(s)}\geq 1, t\geq k\right\}\right|=\min\left\{n|\ell_{G_1}^{H_{a^{k}+1}}\Big(\frac{1-h^+(s)}{1-h^-(s)}\Big)^n\frac{h^+(s)}{h^-(s)}< 1\right\}.
\end{align*}

We know that $\ell_{G_1}^{H_{a^{k-1}+1}}\frac{h^+(s)}{h^-(s)}>1$ because $m^k=s$ and $m^{k-1}=s-1$, so taking logarithms we obtain
\begin{align*}
    	&\quad \left|\left\{m^t=s|\ell_{G_1}^{H_{a^t+1}}\frac{h^+(s)}{h^-(s)}\geq 1, t\geq k\right\}\right|=\min\left\{n|n \ln\Big(\frac{1-h^-(s)}{1-h^+(s)}\Big)> \ln\Big(\frac{h^+(s)}{h^-(s)}\Big)+\ln(\ell_{G_1}^{H_{a^{k}+1}})\right\} \\
    	&\geq\min\left\{n|n \ln\Big(\frac{1-h^-(s)}{1-h^+(s)}\Big)> \ln\Big(\frac{h^+(s)}{h^-(s)}\Big)+\ln(\ell_{G_1}^{H_{a^{k}+1}})-\ln(\ell_{G_1}^{H_{a^{k-1}+1}})-\ln\Big(\frac{h^+(s-1)}{h^-(s-1)}\Big)\right\}   	
    	\end{align*}

Using the recursive form in \eqref{eqn:ldecrease} once again, we get
\begin{align*}
	&\quad \left|\left\{m^t=s|\ell_{G_1}^{H_{a^t+1}}\frac{h^+(s)}{h^-(s)}\geq 1, t\geq k\right\}\right|       	\\
       	&\geq \min\left\{n|n \ln\Big(\frac{1-h^-(s)}{1-h^+(s)}\Big)> \ln\Big(\frac{h^+(s)h^-(s-1)}{h^-(s)h^+(s-1)}\Big)+\ln\Big(\frac{1-h^+(s-1)}{1-h^-(s-1)}\Big)\right\} \\
        &\geq \min\left\{n|(n+1) \ln\Big(\frac{1-h^-(s)}{1-h^+(s)}\Big)> \ln\Big(\frac{h^+(s)h^-(s-1)}{h^-(s)h^+(s-1)}\Big)\right\} 	
    	\end{align*}
From this point onwards, we need to compute a lower bound of $\ln\Big(\frac{h^+(s)h^-(s-1)}{h^-(s)h^+(s-1)}\Big)$, and an upper bound of $\ln\Big(\frac{1-h^-(s)}{1-h^+(s)}\Big)$. For ease of reading we put the derivations of the calculations of the lower bound of $\ln\Big(\frac{h^+(s)h^-(s-1)}{h^-(s)h^+(s-1)}\Big)$ in Appendix~\ref{hh-converge}, and the upper bound of $\ln\Big(\frac{1-h^-(s)}{1-h^+(s)}\Big)$ using the Claim \ref{clm3} in Appendix~\ref{pf:clm3} (by using the fact that $\ln\Big(\frac{1-h^-(s)}{1-h^+(s)}\Big)<\ln\Big(\frac{1}{1-h^+(s)}\Big)$ first).

	Combining the result in Appendix\ref{hh-converge}, which shows $\ln\Big(\frac{h^+(s)h^-(s-1)}{h^-(s)h^+(s-1)}\Big)\geq c_1(p)$, where $c_1(p)$ is only a function of $p$, and the result in Appendix~\ref{pf:clm3}, $\ln(\dfrac{1-h^-(s)}{1-h^+(s)})\leq p^s+p^{2s}$, we can pick $\underbar{n}_s$ as
	\begin{eqnarray}
	\underbar{n}_s= \frac{\ln(c_1(p))}{p^s+p^{2s}}.  
	\end{eqnarray}

\subsection{Calculation of the probability that a wrong cascade will be stopped}\label{subsec:wc_stop}

 First, we use the lower bound of $h^+(m)$ derived in Appendix~\ref{hlowerbound} to get the first inequality
  \begin{eqnarray}
  \mathbb{P}(\text{A wrong cascade can be stopped}) =1-\Pi_{k=1}^{\infty} (1-h^+(m^k)) \geq  1-\Pi_{k=1}^{\infty} (1-p^{m^k}). \nonumber
  \end{eqnarray} 
  Now, we use the sequence $\underbar{n}_s$ that upper bounds the growth rate of $m^k$ to replace the product and apply the lower-bound result in \ref{hh-converge} for large enough $\underbar{J}$. 
  \begin{eqnarray*}
  \mathbb{P}(\text{A wrong cascade can be stopped})&\geq& 1-A(p)\Pi_{n=\underbar{J}}^{\infty} (1-p^n)^{\frac{\ln(c_1(p))}{p^{n-1}(1+p^{n-1})}} \\
    &\geq& 1-A(p)\exp\Big(p\sum_{n=0}^{\infty} \frac{\ln(c_1(p))}{p^n(1+p^{n-1})}(-p^n)\Big) \\
  &=&  1-A(p)\exp\Big(-p\ln(c_1(p))\sum_{n=0}^{\infty} \frac{1}{1+p^n}\Big) =1,
    \end{eqnarray*} 
  where $A(p)=\Pi_{k=1}^{\min\{j|m^{j+1}=\underbar{J}\}} (1-p^{m^k})$.

\subsection{Lower bound of growth rate of $m^k$} \label{lowns}
	Finding the lower bound of $m^k$ is analogous to finding a sequence of $\bar{n}_s$ such that $|\{m^k=s|\ell_{G_1}^{H_{a^{k-1}+1}}\frac{h^+(s)}{h^-(s)}\geq 1\}|\geq \bar{n}_s$ for all $s\in \mathbb{N}$. In contrast to what we did in computing $\underbar{n}_s$, we now want to upper bound the following equation:
	
\begin{align*}
\left|\left\{m^t=s|\ell_{G_1}^{H_{a^t+1}}\frac{h^+(s)}{h^-(s)}\geq 1, t\geq k\right\}\right|=\min\left\{n|\ell_{G_1}^{H_{a^{k}+1}}\Big(\frac{1-h^+(s)}{1-h^-(s)}\Big)^n\frac{h^+(s)}{h^-(s)}< 1\right\}.
\end{align*}

Given that $m^k=s$ and $m^{k-1}=s-1$, $\ell_{G_1}^{H_{a^{k}+1}}\frac{h^+(s-1)}{h^-(s-1)}<1$ because $m^k=s$ and $m^{k-1}=s-1$, we can lower bound the equation and then take logarithms on both sides of the inequality to get
\begin{align*}
    	&\quad \left|\left\{m^t=s|\ell_{G_1}^{H_{a^t+1}}\frac{h^+(s)}{h^-(s)}\geq 1, t\geq k\right\}\right|=\min\left\{n|n \ln\Big(\frac{1-h^-(s)}{1-h^+(s)}\Big)> \ln\Big(\frac{h^+(s)}{h^-(s)}\Big)+\ln(\ell_{G_1}^{H_{a^{k}+1}})\right\} \\
    	&\leq \min\bigg\{ n|n \ln\Big(\frac{1-h^-(s)}{1-h^+(s)}\Big)> \ln\Big(\frac{h^+(s)}{h^-(s)}\Big)+\ln(\ell_{G_1}^{H_{a^{k}+1}})-\ln(\ell_{G_1}^{H_{a^{k}+1}})-\ln\Big(\frac{h^+(s-1)}{h^-(s-1)}\Big)\bigg\} \\
    	&=\min\Big\{n|n \ln\Big(\frac{1-h^-(s)}{1-h^+(s)}\Big)> \ln\Big(\frac{h^+(s)h^-(s-1)}{h^-(s)h^+(s-1)}\Big)\Big\} 	    	
    	\end{align*}
	Similarly, we can replace $ \ln\Big(\frac{h^+(s)h^-(s-1)}{h^-(s)h^+(s-1)}\Big)$ by an upper bound $c_2(p)$ derived in Appendix \ref{hh-converge} for large enough $s$. However, unlike what we did for upper bounding the growth rate of $m^k$, now we keep the form $ \ln\Big(\frac{1-h^-(s)}{1-h^+(s)}\Big)$ and bring this into the calculation of the probability that a right cascade will being stopped. In other words, $$\bar{n}_s=\frac{\ln(c_2(p))}{\ln\big(\frac{1-h^-(n)}{1-h^+(n)}\big)}.$$

  \subsection{Calculation the probability that a right cascade will finally being stopped}\label{hh-upperbd}
	
	Taking the lower bound derived in Appendix~\ref{lowns}, we get the following inequality:
	\begin{align*}
	 \mathbb{P}(\text{A right cascade will be stopped})
	  &=1-\Pi_{k=1}^{\infty} (1-h^-(m^k)) \\
  &\leq 1-A(p)\Pi_{n=\bar{J}}^{\infty} (1-h^-(n))^{\frac{\ln(c_2(p))}{ln\big(\frac{1-h^-(n)}{1-h^+(n)}\big)}} \\
  &= 1-A(p)\exp\Big(\sum_{n=\bar{J}}^{\infty} \frac{\ln(c_2(p))}{\ln\big(\frac{1-h^-(n)}{1-h^+(n)}\big)}\ln(1-h^-(n))\Big),
	\end{align*}
	  where $A(p)=\Pi_{k=1}^{\min\{j|m^{j+1}=\bar{J}\}} (1-h^-(k))$, and $\bar{J}$ is the same as in Appendix \ref{subsec:wc_stop}.
	
	Now, the rest of calculations are just using well-known upper and lower bounds on the logarithm. With calculations detailed in \ref{boundrightcascade}, we get that
  \begin{align*}
 	\mathbb{P}(\text{A right cascade will be stopped})\leq  1-B(p)\exp\Big(\sum_{n=\bar{J}}^{\infty}\dfrac{1}{
    c_4(p)c_1(p)^n-1}\Big),
    \end{align*}
where $B(p)=c_2(p)e^{c_3(p)}\Pi_{k=1}^{\min\{j|m^{j+1}=\bar{J}\}} (1-h^-(k)).$

Now, using the fact that $c_1(p)> 1$ in Appendix \ref{hh-converge}, for a large enough $n$, $|\tfrac{c_4(p)c_1(p)^n-1}{c_4(p)c_1(p)^{n+1}-1}-c_1(p)|\leq \epsilon$. Hence, by ratio test, we can claim that  $\sum_{n=\bar{J}}^{\infty}\tfrac{1}{c_4(p)c_1(p)^n-1}$ converges, and 

$\mathbb{P}(\text{A right cascade will be stopped})<1$.

\subsection{Technical Claims and Calculations}
\subsubsection{Claim 2 and its proof} \label{pf:clm2}
\begin{clm}	
	The likelihood ratio $\ell^{H_{a^k+1}}_{G_1}$ is a strictly decreasing function of $k$.
\end{clm}
\begin{proof}
First, $\frac{h^+(m)}{h^-(m)}$ is a strictly increasing function of $m$ (please see calculations in Appendix~\ref{hh-converge}). Since $h^+(1)>h^-(1)>0$ and given the form in \eqref{eq:lklratio}, $\ell^{H_{a^k+1}}_{G_1}$ is an strictly decreasing function.
\end{proof}

\subsubsection{Claim 3 and its Proof}\label{pf:clm3}

\begin{clm} \label{clm3}
$\ln\left(\dfrac{1-h^-(s)}{1-h^+(s)}\right)$ can be upper-bounded by $p^s+p^{2s}$.
\end{clm}
\begin{proof}
First, $\ln\left(\dfrac{1-h^-(s)}{1-h^+(s)}\right)\leq \ln\left(\dfrac{1}{1-h^+(s)}\right)$.

Then, using the upper bound in Appendix \ref{hlowerbound} and Taylor's expansion, we know that $$\ln\left(\dfrac{1}{1-h^+(s)}\right) \leq -\ln(1-p^s) \leq p^s+p^{2s}.$$ 
\end{proof}	

\subsubsection{Calculation of Upper bound of $h^+(m)$}\label{hlowerbound}
First, in this question guidebook, $h^+(m)$ has a closed form given by
\begin{eqnarray}
h^+(m)=\left\{
\begin{array}{l l}	
\sqrt{\frac{p}{4-3p}}2^{-(m+1)}[(p+\sqrt{p(4-3p)})^{m+1}-(p-\sqrt{p(4-3p)})^{m+1}],& m \text{ is even} \\
\sqrt{\frac{p}{4-3p}}\Big(\dfrac{p}{2}\Big)^{\frac{m+1}{2}}\Big[\Big((2-p)+\sqrt{p(4-3p)}\Big)^{\frac{m+1}{2}}-\Big((2-p)-\sqrt{p(4-3p)}\Big)^{\frac{m+1}{2}}\Big],& m \text{ is odd} \nonumber
\end{array}\right.
\end{eqnarray}

With the closes form of $h^+(m)$, a simple bound of $\dfrac{h^+(m+1)}{h^+(m)}<p$ holds for all $m\in \mathbb{N}$.

For readers curious about the exact difference of $ph^+(m)-h^+(m+1)$, the form is \\
$p^{m+2} \Big[{}_2F_1(-\frac{m}{2}-\frac{1}{2},-\frac{m}{2},-m-1,4(1-\frac{1}{p}))-{}_2F_1(-\frac{m}{2}+\frac{1}{2},-\frac{m}{2},-m-1,4(1-\frac{1}{p}))\Big]$, where ${}_2F_1(\cdot)$ is a hypergeometric function. We will not need this detail as getting the bound $\dfrac{h^+(m+1)}{h^+(m)}<p$ is good enough for our results.

\subsubsection{Convergence and bounds of $\dfrac{h^+(m+1)}{h^-(m+1)}/\dfrac{h^+(m)}{h^-(m)}$.}\label{hh-converge}
For the simplicity, suppose $m$ is even.
Recall the closed form of $h^+(m)$ and $h^-(m)$ as follows:
\begin{eqnarray*}
&&h^+(m)=
\sqrt{\frac{p}{4-3p}}2^{-(m+1)}[(p+\sqrt{p(4-3p)})^{m+1}-(p-\sqrt{p(4-3p)})^{m+1}] \\
&&h^-(m)=
\sqrt{\frac{1-p}{1+3p}}2^{-(m+1)}[(1-p+\sqrt{1+2p-3p^2})^{m+1}-(1-p-\sqrt{1+2p-3p^2})^{m+1}] 
\end{eqnarray*}

Consider $\dfrac{h^+(m+2)}{h^-(m+2)}/\dfrac{h^+(m)}{h^-(m)}$.
\begin{eqnarray}
\dfrac{h^+(m+2)}{h^-(m+2)}\dfrac{h^-(m)}{h^+(m)} 
&=&\dfrac{(p+\sqrt{p(4-3p)})^{m+3}-(p-\sqrt{p(4-3p)})^{m+3}}{(1-p+\sqrt{1+2p-3p^2})^{m+3}-(1-p-\sqrt{1+2p-3p^2})^{m+3}}  \nonumber \\
&\times &
\dfrac{(1-p+\sqrt{1+2p-3p^2})^{m+1}-(1-p-\sqrt{1+2p-3p^2})^{m+1}}{(p+\sqrt{p(4-3p)})^{m+1}-(p-\sqrt{p(4-3p)})^{m+1}} \nonumber \\
&=& \dfrac{(p+\sqrt{p(4-3p)})^2}{(1-p+\sqrt{1+2p-3p^2})^2}
\dfrac{1-\frac{(p-\sqrt{p(4-3p)})^{m+3}}{(p+\sqrt{p(4-3p)})^{m+3}}}{1-\frac{(p-\sqrt{p(4-3p)})^{m+1}}{(p+\sqrt{p(4-3p)})^{m+1}}}
\dfrac{1-\frac{(1-p-\sqrt{1+2p-3p^2})^{m+1}}{(1-p+\sqrt{1+2p-3p^2})^{m+1}}}{1-\frac{(1-p-\sqrt{1+2p-3p^2})^{m+3}}{(1-p+\sqrt{1+2p-3p^2})^{m+3}}}\nonumber \\
&\leq& \dfrac{(p+\sqrt{p(4-3p)})^2}{(1-p+\sqrt{1+2p-3p^2})^2}\Bigg[1+\Big(\frac{p-\sqrt{p(4-3p)}}{p+\sqrt{p(4-3p)}}\Big)^{m+1}+\Big(\frac{p-\sqrt{p(4-3p)}}{p+\sqrt{p(4-3p)}}\Big)^{m+2}\nonumber \\
&&-\Big(\frac{1-p-\sqrt{1+2p-3p^2}}{1-p+\sqrt{1+2p-3p^2}}\Big)^{m+1}-\Big(\frac{1-p-\sqrt{1+2p-3p^2}}{1-p+\sqrt{1+2p-3p^2}}\Big)^{m+2} \Bigg] \nonumber 
\end{eqnarray}

Since $\dfrac{h^+(m+2)}{h^-(m+2)}\dfrac{h^-(m)}{h^+(m)}$ converges to $\dfrac{(p+\sqrt{p(4-3p)})^2}{(1-p+\sqrt{1+2p-3p^2})^2}$ exponentially fast, for large enough $m$, we can bound $\left|\dfrac{h^+(m+2)}{h^-(m+2)}\dfrac{h^-(m)}{h^+(m)} - \dfrac{(p+\sqrt{p(4-3p)})^2}{(1-p+\sqrt{1+2p-3p^2})^2}\right|<\epsilon$. 

Given the fact that $\dfrac{(p+\sqrt{p(4-3p)})^2}{(1-p+\sqrt{1+2p-3p^2})^2}>1$ for all $p>0.5$, there exists functions $c_1(p),c_2(p)>1$ such that $c_1(p)<\dfrac{h^+(m+2)}{h^-(m+2)}\dfrac{h^-(m)}{h^+(m)}$ for all $m$ greater than some natural number $\underbar{J}$ and $\dfrac{h^+(m+2)}{h^-(m+2)}\dfrac{h^-(m)}{h^+(m)}<c_2(p)$ for all $m$ greater than some natural number $\bar{J}$.

\subsubsection{Bounding the probability of stopping a right cascade} \label{boundrightcascade}

  \begin{align*}
  &\qquad \mathbb{P}(\text{a right cascade will be stopped})\\
  &\leq  1-A(p)\exp\Big(\ln(c_2(p))\sum_{n=\bar{J}}^{\infty}\dfrac{-h^-(n)}{\ln(1+\frac{h^+(n)-h^-(n)}{1-h^+(n)})}\Big) \\
    &\leq  1-A(p)\exp\Big(\ln(c_2(p))\sum_{n=\bar{J}}^{\infty}\dfrac{-h^-(n)}{\frac{h^+(n)-h^-(n)}{1-h^+(n)}(1-\frac{h^+(n)-h^-(n)}{1-h^+(n)})}\Big) \\
    &\leq  1-A(p)\exp\Big(\ln(c_2(p))\sum_{n=\bar{J}}^{\infty}\dfrac{-1}{
    (\frac{h^+(n)}{h^-(n)}-1)(1-\frac{h^+(n)-h^-(n)}{1-h^+(n)})\frac{1}{1-h^+(n)}}\Big) \\    
    &\leq  1-A(p)\exp\Big(\ln(c_2(p))\sum_{n=1}^{\infty}\dfrac{-1}{
    (\frac{h^+(n)}{h^-(n)}-1)(1-\frac{p^n-0}{1-p^n})\frac{1}{1-p^n}}\Big) \\   
    &\leq  1-A(p)\exp\Big(\ln(c_2(p))\sum_{n=1}^{\infty}\dfrac{-1}{
    (\frac{h^+(n)}{h^-(n)}-1)\frac{1-2p^n}{(1-p^n)^2}}\Big) \\  
    &\leq  1-A(p)\exp\Big(\ln(c_2(p))\sum_{n=1}^{\infty}\dfrac{-1}{
    (\frac{h^+(n)}{h^-(n)}\frac{1-2p^n}{(1-p^n)^2}-1}\Big) \\  
    &\leq  1-B(p)\exp\Big(\sum_{n=\bar{J}}^{\infty}\dfrac{1}{
    c_4(p)c_1(p)^n-1})\Big),
    \end{align*}
where $B(p)=c_2(p)e^{c_3(p)}\Pi_{k=1}^{\min\{j|m^{j+1}=\bar{J}\}} (1-h^-(k))$
\section{Appendix: Discussions}
\subsection{Question guidebooks are delicate}\label{sec:delicate}
    For the model described in Section \ref{sec:Problem} assume that every agent is only allowed to ask one binary question to the agent just in front of her. In other words, the capacity of the channel is exactly one bit. We will attempt to manually design the questions with the goal being asymptotic learning. Before we proceed, we remind the reader that in our model when indifferent (i.e., the posterior beliefs on the states of world are equally likely, or alternatively the likelihood ratio of state is $A$ versus the state is $B$ is $1$), an agent will always take the action suggested by her private signal.
    
   It is clear that along a sample-path, the only interesting parts are the questions asked after a cascade happens. Therefore, without loss of generality, we assume the first two agents get private signal $A$ and take action $\bar{A}$ to start a $\bar{A}$ cascade.
    
    Give the above assumption, it is clear that agent $3$ cannot stop the cascade as she exactly knows the private signal of agent 1 and 2 from their actions. Hence, agent 3 will take action $\bar{A}$ whichever private signal she receives. Now, agent 4 has a chance to stop the cascade if both she and agent 3 get $B$. Therefore, once she receives signal $B$, she should definitely ask agent 3 the question ``Did you get signal $B$?" Now we have another issue on designing the question book. What questions should agent 4 ask if she gets signal $A$. Since the above question is the most informative question that agent 4 can ask, we assume she will ask the same question even if she gets signal $A$.
    
    When the game goes to agent 5, there are two possible observations for her, either $\bar{A}\bar{A}\bar{A}\bar{B}$ or $\bar{A}\bar{A}\bar{A}\bar{A}$. If $\bar{A}\bar{A}\bar{A}\bar{B}$ is observed, the she definitely know that the underlying private signal is $AABB$ and the cascade has 
    already been stopped. Hence, she should take a decision according to her private signal. The interesting part comes when she observes $\bar{A}\bar{A}\bar{A}\bar{A}$. Now, prior to asking any questions, she knows there are at least three $A$ in the first four agent's private signal. Therefore, she should take action $\bar{A}$ whatever signal she receives. However, her question, even though it cannot benefit herself, may be critical to the following agents and could enable them to stop the cascade. Suppose she gets $\bar{B}$, and she can ask agent 4 the question ``Did you or agent 3 get signal $\bar{B}$?". First, we realize that agent 4 is able to answer this question. Once she get positive answer, she know there are exactly 3$\bar{A}$s and 2 $\bar{B}$s received by the first 5 agents. Therefore, once the agent 6 ask if there are two $\bar{B}$s in the first five agents, she can answer this question, and the positive answer can help the agent 6 to stop the cascade once she gets $\bar{B}$. Now, we have a clear example how a silent agent can actually help the following active agents to stop the cascade by asking questions in preparation of the (possible) questions asked by active agents. 
     
     Everything works well until agent 6. However, once agent 6 gets signal $A$, there exist multiple questions that are informative. For example, agent 6 can ask question 1: ``Did the first five agents get 2 $B$s?" or question 2: ``Did the first five agents get one or more $B$s?". Both questions can be answered by agent 5, but it is hard for us to know which one is the better question to help achieving asymptotic learning. Furthermore, agent 6 cannot distinguish the exact number of $B$ in the first 5 agents once she receives $B$ but gets a negative answer from agent 5. Therefore, when the cascade continuous to agent 7, we have to keep every branch of possible questions alive, and try to proceed on every branch until we finally understand that some branches cannot achieve asymptotic learning. Therefore, we can realize that question guidebooks are very delicate to design even we have already restrict our attention to a single Yes-No question. To conclude, even though we can proceed and design question guidebooks in this way, it is hard to analyze and generalize the result to non-binary states and channels with higher capacity. 
The main value of this example is to indicate that it may not be the best approach to design the question guidebooks directly, and also to justify why we want to work with the corresponding Markov chains instead.  
     
\subsection{Why learning is more difficult in a deterministic network topology} \label{detvsranNW}

Unlike a deterministic network where the topology is the common knowledge of every agent, a common assumption is to use a randomized network topology where it is assumed that distributions of edges are common knowledge but the realized set of predecessors $\mathcal{B}_n=\{m|(m,n)\in E\}$ that agent $n$ can communicate is a private information of agent $n$ (otherwise higher order beliefs will be hard to analyze). Under this assumption, if an information cascade occurs and the information that is allowed to get from predecessors with communication channels is restricted to their private signals, it is analogous to the model in \cite{acemoglu2011bayesian} with neighbor set $N_n=\mathcal{B}_n \cup \{i\} \cup \{i+1\}$, where agent $i,i+1$ are the agents to start the cascade. Then, Theorem 4 in \cite{acemoglu2011bayesian} states the condition that asymptotic learning can be achieved in randomized networks. In contrast, things are not that simple with deterministic networks as there is an informational Braess's paradox, i.e., two information sources that both individually suggest agent $n$ to ignore her private signal and take action $\bar{A}$ may collectively lead agent $n$ take action $\bar{B}$ instead. An example will be provided below to illustrate this phenomenon, but the main message here is that problems in deterministic and randomized network topology are substantially different.

\subsubsection{Informational Braess's paradox in deterministic networks}
	Here, we want to present a counter-intuitive example showing that informational monotonicity may not exist in the deterministic network. In a nutshell, two (or more) sources of information both suggest an agent ignoring her private signal and taking an action $\bar{A}$ may eventually makes agent taking action $\bar{B}$.
	
     Consider the observable history is represented by network $\mathcal{G}$ contains two subnetwork $G_1(U,E_U)$ and $G_2(V,E_V)$ and two nodes $W_1, W_2$. Let $e^U_{i,j}$ represents the edge between $U_i$ and $U_j$, and similarly for $e^V_{i,j}$.
     
     Now, define the topology of the first subnetwork $G_1(U,E_U)$ such that $E_U=\{e^U_{1,3}, e^U_{2,3}, e^U_{3,j} \forall j\in \{[J]\setminus \{1,2,3\}\} \}$, where $J$ is a constant integer; and define the topology of the second subnetwork $G_2(V,E_V)$ such that $E_V=\{e^V_{1,3}, e^V_{2,3}, e^V_{1,k}, e^V_{3,k} \forall k\in \{[K]\setminus \{1\},\{2\},\{3\}\} \}$, where $K$ is a fixed integer.
     
     Then, the whole network $\mathcal{G}$ is defined as follows: $\mathcal{G}=(U \cup V \cup \{W_1,W_2\}, E_U \cup E_V \cup E_{W_2})$, where $E_{W_2}=\{(i,W_2)| i\in U \cup V\cup W_1 \setminus \{U_2,V_1,V_2\}\}$. In short, $W_2$ can observe all agents action except agent $U_2$, $V_1$, and $V_2$. We assume that agents make actions sequentially, and agent $X_i$ takes action before agent $X_j$ for all $i<j$, $X\in \{U,V,W\}$ and agent $X_i$ takes action before agent $W_2$ for all $X\in \{U,V\}$. The directed graph is depicted in Figure~\ref{fig1:BP}.     
         \begin{figure}[H] \label{fig2}
    \centering
	\includegraphics[width=0.65\linewidth]{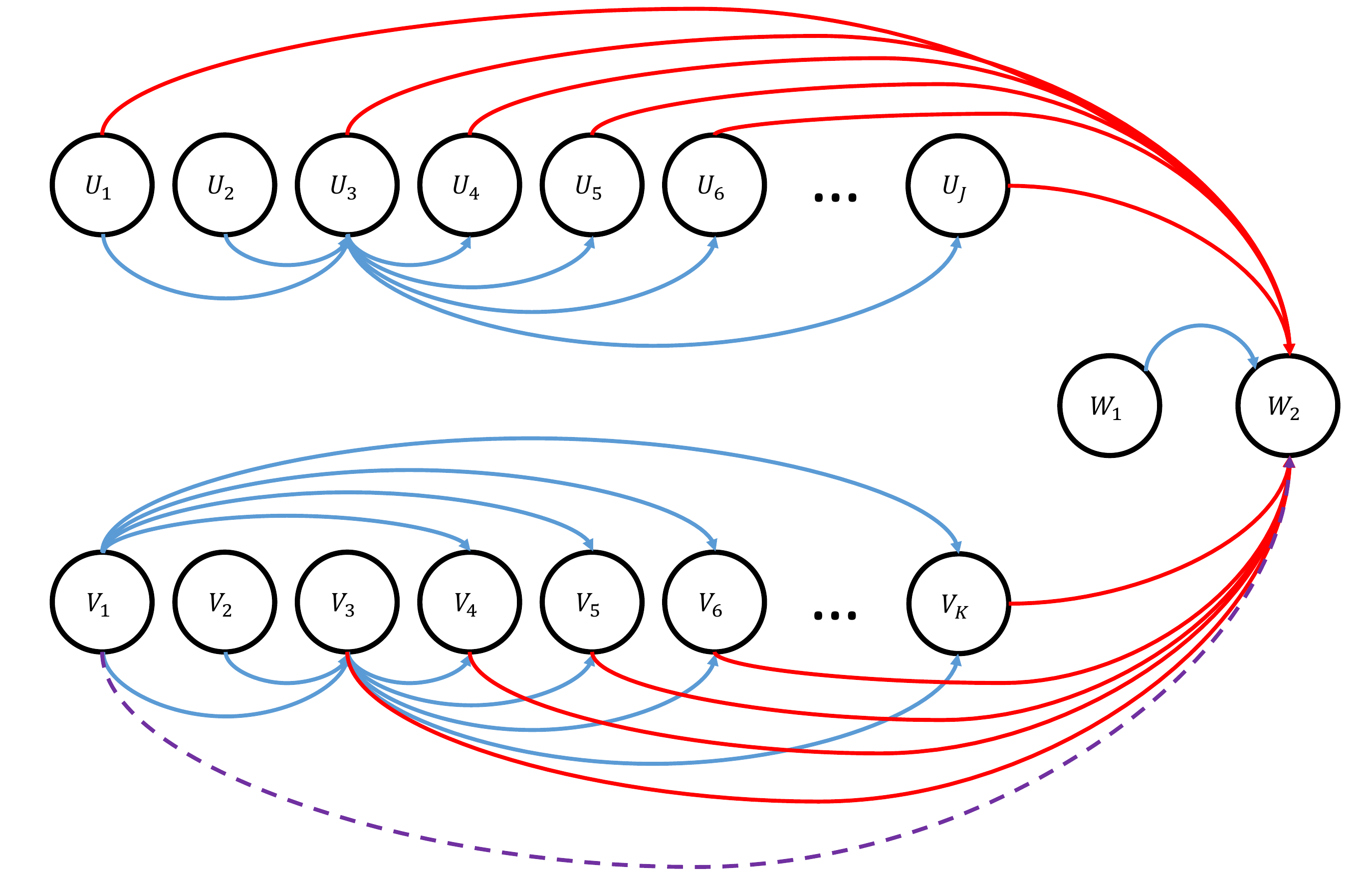}
    \caption{Baress's paradox in deterministic network}
    \label{fig1:BP}
    \end{figure} \vspace{-12pt} 
	In this topology, consider a realized observation of $W_2$ such that actions taken by agents in $U_j$, $X_{U_j}=\bar{A}$ for all $j\geq 3$ and $j=1$, and actions taken by agents in $V_k$, $X_{V_k}=\bar{B}$ for all $k\geq 3$. By calculating the likelihood ratio at $W_2$, we know that $W_2$ should ignore her private signal and take action $\bar{A}$ the observation of $X_{W_1}=\bar{A}$.
     
     At this point, we know that if agent $W_2$ observes the  history realized in the above case, there is a information cascade on $\bar{A}$; and we say this is the first source of information that agent $W_1$ can get from.
     
     Now, if on top of these observations, agent $W_2$ now can observe the action taken by agent $V_1$ (the purple dash line in the figure) and $X_{V_1}=\bar{A}$, which solely suggests agent $W_2$ to take action $\bar{A}$. Surprisingly, if we calculate the likelihood ratio of $W_2$ using the observation of both source of information, agent $W_2$, instead of being in an information cascade and taking action $\bar{A}$, will ignore her private signal and taking action $\bar{B}$ when $K\geq 5$. In this case, although both sources of information suggest an information cascade, by considering the feasible sequence of the private signal using the both source of information, the agent can choose to initiate an information cascade opposite to the suggested one.  
         
     Instead of presenting the detailed calculations of the likelihood ratio, we want to point out the high-level idea why such an informational Braess's paradox can happen in deterministic networks.
     
	In deterministic networks, agents know the whole network topology, and can build high-order belief conditional on the past agent's topology to get some information about some predecessors' action they cannot observe directly. In the example shown above, before observing the action taken by $V_1$, agent $W_2$ doesn't know if the action $X_{V_k}=\bar{B}$ she observes is taken by agent $V_k$'s private signal or $V_k$ is in an information cascade of $\bar{B}$. Therefore, she need to take both scenarios into consideration to update her posterior belief if $X_{V_1}$ (the purple line in the figure) is not observable. However, once she know that $X_{V_1}=\bar{A}$, she knows that all $X_{V_k},~k\geq 3$ are taken by their private signal. When $K$ is large enough, agent $W_2$ will drop her private signal and initiate an information cascade $\bar{B}$.

\subsection{Extended Question Guidebooks} \label{extQGB}
	Here, we want to point out an extension of the framework of our question guidebook. As mentioned in Section \ref{sec:Problem}, we are only allowed to design question guidebooks suggesting agents asking questions simultaneously. However, assume $m\in \mathcal{B}_n$, questions being asked by agent $n$ to agent $m$ can be allowed to be dependent not only on the private signal and history observed by agent $n$, but also on the answers of questions asked to other agents in $\mathcal{B}_n$ prior to asking agent $m$. In short, the order of agents in $\mathcal{B}_n$ queried matters in the general framework. The framework of question guidebook, in general, should be able to help agents update their higher-order belief in the process of asking questions, and not just a one-shot update after all responses of questions are received. In the extended framework, a particular collection of sequence of set of questions the information designer provides to agents will be collectively called a \textit{extended question guidebook}. The systematic analysis of extended question guidebooks which includes the freedom of query order on predecessors in $\mathcal{B}_n$, is generally hard because of the combinatorial complexity of considering all the possible order of questions, and because the recursive analysis of the higher order beliefs on each specific order of questions has to be accounted for. To give a glimpse of how an extended question guidebooks can help achieving asymptotic learning efficiently by asking questions to agents with right order, a simple example is provided. 
   
\subsubsection{Example that the order of questions matters}  \label{warm-up}
	Before characterizing question guidebooks, we provide a simple example demonstrating question guidebooks in the general extended framework, where a question guidebook is a collection of a sequence of an ordered set of questions. 
	
	We start with the most unrestricted network, all prior agents are contactable, so every agent is free to ask questions to any agents with lower index than her (if any), and there are no capacity constraints on these communications. Obviously, every agent can get perfect information by asking enough questions to prior agents so that asymptotic learning can be achieved. In the most naive case, every agent just asks all past agents their private signal. This will require $\Theta(n)$ bits of information per agent.
	 
	 Since agents are homogeneous, the learning can be achieved with less communication through the ``backward level tracking" scheme proposed below. 
	First of all, it an agent is not in a cascade, she take actions according to her private signal.	 
	For agents in an cascade, suppose they are in $A$ cascade, there are three different cases detailed below:
	If she gets $A$, then she plays $\bar{A}$ without asking any questions.
	If she gets $B$, then she asks her immediate predecessor's private signal. If her predecessor's private signal is $A$, then she can safely play $A$. The reason is that she knows that there must be an agent in the history has the same number of $\#A-\#B$ and that person plays $A$. Hereafter, we say $\#A-\#B$ be the \textit{level} of agents.
	If she gets $B$ and her immediate predecessor also got $B$, she can safely play $B$ only if she knows the current level is over $2$. She can ask learn the level by asking agents $n-1, n-2, \ldots$ about their private signal. Assuming that agent plays correctly, he would then know the level of that agent, and from this could compute his own level. 
		
	 The remaining question is how long an agent will need to look back in order to find either the beginning, an agent that plays a different action, or an index agent.  In the long run, if the correct state of the world is $B$ all agents will eventually play $B$.  Thus $B$ agents will simply play $B$ (because the prior agent did).  An $A$ agent will need to look back until he finds an index agent.  However, this is essentially a (downward) biased random walk; and so the expected number of steps until locating an index agent is constant.  However, it is only constant in the average case.  If the graph of the level goes down for $\log(n)$ steps (as is probable at some point), then agent $n$ will have to ask at least $\log n$ prior agents to find an index agent.  This scheme will not work for agents who only ask a fixed number of queries.  Our surprising result is that, actually, agents can manage with only asking just the immediate predecessor a fixed number of queries, namely one.
	 
	 \subsection{A subtle weaker assumption on common knowledge of question guidebook} \label{commknown}
	 At the end of the Section \ref{sec:Problem}, we assume that the whole question guidebook is the common knowledge. However, it is no loss to use a slightly weaker assumption that only \textbf{\textit{the feasibility and the compatibility of the proposed question guidebook}} is common knowledge. To see the difference between this two assumptions, we have to consider the set of question guidebooks containing index-dependent questions.
	 
	 Now, consider we have two different question guidebooks, $Q_1$, $Q_2$ both achieves asymptotic learning. These two questions guidebooks, concerning the usefulness of questions, become operational only when a cascade is initiated. Now, consider a joint question guidebook $Q^*$ which uses questions in $Q_1$ when a cascade is initiated by the agent with odd index, and use questions in $Q_2$ when cascade is initiated by even-indexed agent. (This question guidebook can still achieve asymptotic learning, but it's not the main point.) Applying the original assumption, we need to disclose the whole question guidebook to all agents, but it is straightforward that a half of the question guidebook, after an agent observes the history and knows that the cascade was initiated by an odd-indexed agent or an even-indexed one, is useless for her on updating her posterior beliefs. By restricting to the assumption which commits to the feasibility and incentive compatibility of the question guidebook, all agents' best response are to ask the questions suggested by the question guidebook, and this is enough to make our scheme stand.

\subsection{Complementary discussion on mapping question guidebook to Markov chains} \label{QGBtoMC}

To formulate the mapping, we need to know how many of states are required in these Markov chains first. 
   
   Since each channel $e$ has finite capacity, denoted by $c(e)$ and $|B(n)|<\infty$, the answers of questions asked by agent $n$ conditioned on the observed history and her private signal can partition the information space to $q_n(H_n,s_n)$ disjoint information sets, where $q_n(H_n,s_n) \leq \Pi_{e\in B(n)} 2^{c(e)}$. Besides, the minimum number of states sufficient to represent $\mathcal{I}_n(Q,H_n)$ in agent $n$ also depends on the number of information sets in the agent $n-1$, $|\mathcal{I}_{n-1}(Q,H_{n-1})|$. Thus, the upper bound of the minimum number of states sufficient to represent $\mathcal{I}_n(Q,H_n)$ can be written as a recursive form:
       \begin{equation} \label{eqn:1} \vspace{-6pt}
         |\mathcal{I}_n(Q,H_n)|=\min \{|\mathcal{I}_{n-1}(Q,H_{n-1})|, q_n(H_{n-1},A)\}+\min \{|\mathcal{I}_{n-1}(Q,H_{n-1})|, q_n(H_{n-1},B)\} 
   \end{equation}
	
     With the knowledge of the number of required states in Markov chains, next we associate a question guidebook with a sequence of sets of Markov chains sharing the same state space. Denote $(G,(\mathcal{M}_n))$ to be a sequence of sets of Markov chains, where $G$ is the set of states and $\mathcal{M}_n$ is the set of Markov chains at time $n$. Since we are allowed to ask different questions for different observed histories, we can have different Markov chains for different feasible histories according to the current question guidebook, hence $|\mathcal{M}_t|=|\{H_n|P(H_n|Q)>0\}|$.
     With \eqref{eqn:1}, we know that that $|G|=\sup_{n,H_n\in \mathcal{H},s_n} |\mathcal{I}_n(Q,H_n)|$ suffices for our question guidebooks. With such a $G$,  questions asked by agent $n$ conditioned on an observed history $H_n$ and can be written as a unique Markov chain transition matrix. However, for the simplicity of identifying the case that active players actually stop the cascade, an extra state, $G_0$ is added to capture events when this cascade is stopped. For every silent agent, $G_0$ is an isolated state which has no (directed) edges to/from other states. For active agents, once a cascade is stopped, the state transitions to stage $G_0$ irrespective of the previous state. With the new $G^*=G\cup G_0$, once we specify the series of sets of Markov transition matrices that each agent uses for information set partition, the question guidebook is uniquely determined.

\section{Appendix: Related work} \label{App:relatedwork}

Social learning, or so called Bayesian observation learning, studies whether and how consensus (unanimous actions) can be reached among sequential Bayes-rational decision makers under incomplete information. Key result shown in \cite{BHW1992,welch1992,banerjee1992} says that with homogeneous and Bayes-rational agents receiving binary private signals, an \textbf{\textit{Information cascade}}, all but a few of the first agents will cascade to the less profitable action, happens almost surely. Once an information cascade occurs, no future private signals are revealed. Smith and S{\o}rensen \cite{smith2000} significantly generalized\footnote{The work in \cite{sorensen1996rational,smith2000} also developed the generalization to arbitrary but finite states of the world and a finite set of actions.} the model to allow for richer signals characterized by the likelihood ratio of the two states of the world deduced from the private signals, and heterogeneous agent types\footnote{Here a new phenomenon called ``confounded learning" is demonstrated where in the long run agents of different types will herd on the same action and from the actions it will be impossible to detect the types.}. Using both martingale techniques and Markovian analysis, they proved that with unbounded likelihood ratios in the private signals, any cascade can be stopped and learning toward the correct action can be achieved; the speed of learning in this setting is characterized in \cite{hann2018speed,acemoglu2017fast}. 

These initial works lead to considerable follow-up work towards understanding information cascades better, and towards achieving (asymptotic) learning. The vast majority of work here studies how modifying the \textit{information structure} of the problem impacts cascades or allows for learning. As these are closest to our work, we will discuss these in detail to highlight our contributions and the differences. Aside from those works in the above field, there is a vast literature \cite{cover1969hypothesis,hellman1970learning,zhang2013hypothesis,wang2015social,le2017learning}
that consider non-Bayesian agents including bounded rational players, irrational players, and algorithmic agents, and alternate history update rules as a means of achieving learning. A majority of this set of literature studies the (optimal) decision rules in decentralized (binary) hypothesis testing problem on a variety of network models \cite{TayTsiWin:J09a, drakopoulos2016network,tay2008subexponential}. Since the Bayes-rationality constraint differentiates our work from the work on decentralized hypothesis testing problems, we will only discuss several seminal works and highlight one work \cite{drakopoulos2012learning}, in particular, that achievedsasymptotic learning with a specific set of four-state Markov chains. The approach in \cite{drakopoulos2012learning}, from the perspective of information design, is similar in spirit to partitioning information state to information sets, but for non-Bayesian agents. Last, there is also literature that allows for heterogeneous types of agents or changes the \textit{actions and the payoff structure}, and studies the impact of these on social learning. In the literature with heterogeneous agent types\cite{wu2015helpful,vaccari2016social,le2017learning}, disagreement between agents typically leads to an information cascade but the presence of poorly informed agents surprisingly reduces the probability of a wrong cascade. Papers \cite{lee1993convergence,gul1995endogenous,vives1997learning,huck1998informational}  
that modify the actions, typically consider continuum action spaces that result in learning. We will not discuss the last class of literature in detail but include the references for completeness.

As every agent is endowed with a (conditionally) independent and informative private signal about the state of the world, learning would result if these signals are communicated and collected frequently enough and in an accurate fashion. The emergence of information cascades (and herding) shows that the information assimilation part is faulty. Changing the underlying information structure either by revealing only part of the information in the database, by modifying the information in the common database or by adding new channels of information have then been the approaches that have been taken. In \cite{acemoglu2016informational}, only a (random) subset of the past agents' actions are revealed to the current agent. This feature allows some agents to take actions solely based\footnote{By Bayes-rationality, this is equivalent to the revelation of their private signal.} on their private signal, and the paper characterizes the properties of these subsets (called networks) such that learning results. A key result is that even with private signals having a bounded likelihood ratio, there exist networks such that learning occurs. The authors in \cite{mossel2015strategic} show that a special class of simple networks also leads to learning, where agents only see the actions of at most $d$ past agents, and where any agent's action is only observed by agents at most $L$ indices ahead. Recognizing that displaying no past action history would also result in Bayes-rational agents revealing their private signals, \cite{peres2017fragile} 
determined the minimum sequence of revelation agents that are needed for learning: each agent $n$ reveals independently with probability $p_n$ where we need $p_n\propto 1/n$. In \cite{banerjee2004word,smith2008rational} 
, the ordering information of the subset of agents whose actions are revealed, is omitted but nevertheless learning results. Examples of imperfect observation of history that do not lead to complete learning often feature assumptions such as deterministic sub-sampling of past agents' actions (sometimes only by a finite number of agents), unknown observation order and aggregating observations, such as in \cite{sgroi2002optimizing,celen2004observational,callander2009wisdom,ho2014robust,song2015social,bohren2017bounded,chamley1994information,zhang2009robust,sgroi2003irreversible}.

In \cite{le2014value,le2014impact,monzon2017aggregate}, imperfections are added when the information is stored in the database. Whereas this does not result in learning, discontinuous and non-monotonic behavior in the amount of imperfection added is shown; this adds to the literature on information Braess paradoxes observed with equilibrium behavior. Along the same lines, by allowing for stochastic arrivals, in \cite{ISIT2018}, uncertainty in whether an agent arrived and didn't purchase or no agent arrived, results in discontinuous and non-monotonic behavior of the wrong cascade probability in the uncertainty, even though learning doesn't result. Finally, the impact of additional information via reviews of the item obtained when purchases are made, is studied in \cite{le2016quantifying,le2017information,acemoglu2017fast}, with the main conclusions that learning requires unbounded likelihood ratios of the private signals (and not the reviews) and information Braess paradoxes result in extremely non-intuitive behaviors. 

By relaxing the assumption that agents are Bayes-rational, learning, also called optimal decision rule, has been studied in distributed hypothesis testing problems with long history. A seminal work\cite{cover1969hypothesis} studied distributed binary hypothesis testing problem and achieve learning (limit probability of error zero in their language) with agents using a four-valued statistics-based algorithm. A following work by Hellman and Cover \cite{hellman1970learning} showed that asymptotic learning is not achievable under bounded likelihood ratio on signals. General results, without a specified class of network topology, of distributed hypothesis testing were summarized/presented in \cite{tsitsiklis1989decentralized}. Learning in tree networks with bounded depth/degree with algorithmic agents were respectively studied in \cite{TayTsiWin:J09a, drakopoulos2016network}. Restricting attentions to a line-structure network, named tandem network, although there is no learning when the signal distributions are unknown a priori \cite{ho2014robust}, learning can be achieved at the sub-exponential rate for unbounded likelihood ratio in \cite{tay2008subexponential}. In order to study what additional information is required to achieve learning under bounded likelihood ratio, authors in \cite{drakopoulos2012learning} allowed agents to see $K\geq 2$ immediate predecessors actions, instead of just one. It is shown in \cite{drakopoulos2012learning} that asymptotic learning can be achieved using a specific set of four-state Markov chains, with $K=2$. From the perspective of information design, this approach of designing Markov chains for learning is similar in spirit to a partition of information sets, but for non-Bayesian agents.

	Although designing Markov chains for learning, a prototype of information set partition, is provided in \cite{drakopoulos2012learning}, conditions to achieve learning in models with Bayesian agents are usually more harsh than with non-Bayesian agents. A well-known example is that even though neither Bayesian or Non-Bayesian model cannot achieve asymptotic learning in bounded-likelihood ratio signal with binary actions \cite{cover1969hypothesis}, learning under bounded-likelihood ratio signal can be achieved in ternary actions \cite{koplowitz1975necessary} among Non-Bayesian agents, but not among Bayesian agents \cite{dia2009decision}. The Bayesian-agent assumption restricts some particular (realized by nature) agents' action space and hence blocks the information aggregation. To study how to design information structure shown to Bayes-rational agents, there is an emerging field called ``Information design" following a seminal paper ``Bayesian Persuasion" by Kamenica and Gentzkow in \cite{kamenica2011bayesian}. With the commitment power on the information structure (without changing the prior beliefs of agents), designed randomized signal can partition the information space and change agents' posterior beliefs to change agents' preferred actions with some probability. However, the nature of our problem, the finite communication capacity, restricts the strength (richness) of the additional signal available, hence that changing every agents action with additional signals is positive probability in order to stop cascades, is not possible. To overcome this and achieve learning, only a particular set of agents, called active agents in Definition \ref{def:agenttype}, has the positive probability to stop the cascade. The purpose of the rest of agents, called silent agents in Definition \ref{def:agenttype}, is to relay the information to make active agents to be ``persuadable" (have a positive probability to stop the cascade after updating their posterior likelihood ratio). From this perspective, the approach presented in this paper can be viewed as ``relayed Bayesian persuasion". Essentially, a long series of silent agents convince the next active agent that she is in a wrong cascade by making the likelihood ratio become very large, effectively unbounded.

\end{document}